\documentclass[10pt,conference]{IEEEtran}
\usepackage{amsmath,amssymb}
\usepackage{graphicx}
\usepackage{mathrsfs}
\usepackage{tikz}
\usetikzlibrary{arrows}
\tikzset{>=latex'}
\usepackage{booktabs}
\usepackage{nicefrac}
\usepackage{enumerate}
\usepackage{psfrag}	
\usepackage{array}
\usepackage{multirow,hhline}
\usepackage{exscale}
\usepackage{color,balance}
\usepackage{colortbl}
\usepackage{epsfig}
\usepackage{cite}
\usepackage{amsthm}
\usepackage{subfigure}
\definecolor{c1}{rgb}{0,0,0}
\definecolor{c2}{rgb}{0,0,0}
\definecolor{c3}{rgb}{0,0,0}
\definecolor{res}{rgb}{0,0,1} 
\definecolor{rev}{rgb}{0,0,0}

\usepackage{xcolor}

\newsavebox\MBox

\graphicspath{{./Figures/}}

\newtheorem{theorem}{Theorem}

\definecolor{C1}{rgb}{0,0,0}
\newcommand{\X}{\boldsymbol{X}} 
\newcommand{\Y}{\boldsymbol{Y}} 
\definecolor{highlight1}{rgb}{0.95,0.95,0.9}
\definecolor{highlight_S1}{rgb}{0.1,0.5,0.1}
\definecolor{highlight_S2}{rgb}{0.1,0.1,0.5}
\newcommand{\hilightone}[1]{\colorbox{highlight1}{#1}} 
\definecolor{highlight2}{rgb}{0.9,0.95,0.95}
\newcommand{\hilighttwo}[1]{\colorbox{highlight2}{#1}} 

\newcommand{\D}{A}
\newcommand{\I}{D} 
\newcommand{\J}{E} 
\newcommand{\K}{F} 
\newcommand{\CC}{C} 
\newcommand{\DD}{B} 
\newcommand{\KK}{G}

\begin{document}

\IEEEoverridecommandlockouts
\title{Resolving Entanglements in Topological Interference Management with Alternating Connectivity}
\author{
\IEEEauthorblockN{Soheil Gherekhloo, Anas Chaaban, and Aydin Sezgin}
\IEEEauthorblockA{Chair of Communication Systems, 
RUB, Germany\\
Email: { \{soheyl.gherekhloo, anas.chaaban, aydin.sezgin\}@rub.de}}
\thanks{This work is supported by the German Research Foundation, Deutsche Forschungsgemeinschaft (DFG), Germany, under grant SE 1697/10.}
}

\maketitle
\vspace{-1cm}
\begin{abstract}

The sum-capacity of a three user interference wired network for time-varying channels is considered. Due to the channel variations, it is assumed that the transmitters are only able to track the connectivity between the individual nodes, thus only the (alternating) state of the network is known.    
By considering a special subset of all possible states, we show that state splitting combined with joint encoding over the alternating states is required to achieve the sum-capacity.
Regarding upper bounds, we use a genie aided approach to show the optimality of this scheme. 
This highlights that more involved transmit strategies are required for characterizing the degrees of freedom even if the transmitters have heavily restricted channel state information.
\end{abstract}
\section{Introduction}
The optimal management of interference within an interference limited wireless network is a challenging task due to scarce resources such as frequency bandwidth, power, and time. Certainly, the more accurate is the channel state information at the transmitters (CSIT), the more effective is the interference management. However, providing perfect CSIT is a challenging issue in wireless networks, especially for networks with high mobility and size. Due to this, interference management for different setups based on imperfect CSIT attracted the attention of researchers.

In \cite{MaddahaliTse}, the case of completely stale CSIT (using the so-called retrospective interference alignment (IA)) have been considered. For the  multiple-input and multiple-output (MIMO) broadcast channels, they have shown that even completely delayed CSIT can be useful. 

Note that, the quality of CSIT might vary during the overall transmission. To this end, a degrees of freedom (DoF) study based on a mixture of CSIT was addressed in~\cite{GouJafar_mixedCSIT} and~\cite{ShengKobayashiGesbertYi}. Interestingly, it was shown that splitting an alternating CSIT problem into separate setups with fixed CSIT is not optimal. 

As most wireless networks are rather heterogeneous in terms of node mobility and capability, the CSI quality at the transmitters is not the same for all users. This was considered in~\cite{TandonJafarShamaiPoor}, in which users have either perfect, delayed, or no CSIT at all. Similarly, the capacity region of the two-user binary fading channel was characterized in~\cite{VahidMaddahALiAvestimehr} for different models of availability of CSIT.

In \cite{Jafar}, the CSIT is obtained by having a feedback of $1$ bit. Essentially, this feedback provides information about presence or absence of a link. In more detail, it is assumed that a link is absent if its corresponding interference noise ratio (INR) is lower than $1$. Due to this assumption, the information available at the transmitter does not exceed the topology of the network. To this end, the DoF optimal design of networks with $1$ bit feedback is referrred to as ``topological interference management" (TIM)~\cite{Jafar}. Interestingly, it is shown in~\cite{Jafar} that TIM for linear wired and wireless networks reduces to a single problem. In other words, solving one of these problems leads to the solution for the other one, in such a way that the DoF \textcolor{C1}{of} a linear wireless \textcolor{C1}{network leads to} the capacity of the corresponding linear wired channel, \textcolor{c2}{or vice versa}. 
\textcolor{rev}{In~\cite{Jafar2012,NaderializadehAvestimehr2013}, the robustness of topological interference alignment in the fast fading scenario is studied under the assumption that the topology of the network is fixed during the communication.}

While~\cite{Jafar,Jafar2012,NaderializadehAvestimehr2013} considered fixed connectivity within the network, alternating connectivity due to the underlying channel being time-variant was considered in~\cite{SunGengJafar}. It was shown that the DoF for several multi-user channels can only be achieved by joint encoding across \textcolor{C1}{connectivity states}, each having a certain probability of occurrence. The DoF for three user interference channel with two connectivity states was considered in~\cite{SunGengJafar}. \textcolor{rev}{In \cite{GherekhlooChaabanSezginIZS14}, the DoF of a three user interference channel is characterized for the case that} the network has a Wyner-type channel flavor~\cite{Wyner}. 
Note that, the connectivity states were assumed to be equiprobable in~\cite{SunGengJafar} and~\cite{GherekhlooChaabanSezginIZS14}. In this work, we consider the three user interference channel with states of non-equal probabilities. As it turns out, the schemes in~\cite{SunGengJafar, GherekhlooChaabanSezginIZS14} are not enough to characterize the DoF for the general case. In order to highlight this issue, we consider a subset of all possible connectivity states. 
\textcolor{rev}{In more detail, we distinguish between two transmission schemes which use joint encoding: 
\begin{itemize}
\item JESS: Joint encoding across connectivity states and decoding using a single interference resolving state as in~\cite{SunGengJafar, GherekhlooChaabanSezginIZS14}.
\item JEMS: Joint encoding across connectivity states and decoding using multiple interference resolving states.
\end{itemize}}
\newcommand{\startnodesexp}[7]{
\node at (0.5,0) [above] {#1};
\node (t1) at (0,0) [inner sep=0]{};
\node (t1s) at (0,0) [inner sep=0,left]{#2};
\node (t2) at (0,-1) [inner sep=0]{};
\node (t2s) at (0,-1) [inner sep=0,left]{#3};
\node (t3) at (0,-2) [inner sep=0] {};
\node (t3s) at (0,-2) [inner sep=0,left] {#4};
\node (r1) at (2,0) [inner sep=0] {};
\node (r1s) at (2,0) [inner sep=0,right] {#5};
\node (r2) at (2,-1) [inner sep=0] {};
\node (r2s) at (2,-1) [inner sep=0,right] {#6};
\node (r3) at (2,-2) [inner sep=0] {};
\node (r3s) at (2,-2) [inner sep=0,right] {#7};
\draw[->] (t1) to (r1);
\draw[->] (t2) to (r2);
\draw[->] (t3) to (r3);}
\newcommand{\redpathexp}[2]{\draw[dashed,->,red] (#1) to (#2);}
\begin{figure} 
\centering
\begin{tikzpicture}
\startnodesexp{}{Tx~1}{Tx\,2}{Tx\,3}{Rx\,1}{\,2}{Rx\,3}
\redpathexp{t1}{r2}
\redpathexp{t1}{r3}
\redpathexp{t2}{r1}
\redpathexp{t2}{r3}
\redpathexp{t3}{r1}
\redpathexp{t3}{r2}
\end{tikzpicture}
\caption{Three Tx want to communicate with their desired Rx over a time-variant channel. 
The desired links (solid lines) exist always. However, the interference links (dashed lines) are not always present.}
\label{fig:System_model}
\end{figure}
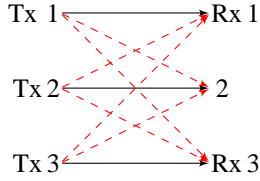
\newcommand{\startnodes}[1]{
\node at (0.5,0) [above] {#1};
\node (t1) at (0,0) [inner sep=0] {};
\node (t2) at (0,-0.4) [inner sep=0] {};
\node (t3) at (0,-0.8) [inner sep=0] {};
\node (r1) at (1,0) [inner sep=0] {};
\node (r2) at (1,-0.4) [inner sep=0] {};
\node (r3) at (1,-0.8) [inner sep=0] {};
\draw[->] (t1) to (r1);
\draw[->] (t2) to (r2);
\draw[->] (t3) to (r3);}
\newcommand{\redpath}[2]{\draw[->,red] (#1) to (#2);}
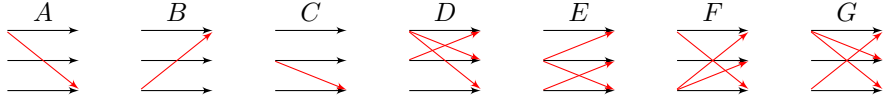
\begin{figure*} 
\centering
\begin{tikzpicture}
\startnodes{$\D$}
\redpath{t1}{r3}
\end{tikzpicture}
\hspace{5mm}
\begin{tikzpicture}
\startnodes{$\DD$}
\redpath{t3}{r1}
\end{tikzpicture}
\hspace{5mm}
\begin{tikzpicture}
\startnodes{$\CC$}
\redpath{t2}{r3}
\end{tikzpicture}
\hspace{5mm}
\begin{tikzpicture}
\startnodes{$\I$}
\redpath{t1}{r3}
\redpath{t2}{r1}
\redpath{t1}{r2}
\end{tikzpicture}
\hspace{5mm}
\begin{tikzpicture}
\startnodes{$\J$}
\redpath{t2}{r3}
\redpath{t2}{r1}
\redpath{t3}{r2}
\end{tikzpicture}
\hspace{5mm}
\begin{tikzpicture}
\startnodes{$\K$}
\redpath{t1}{r3}
\redpath{t3}{r1}
\redpath{t3}{r2}
\end{tikzpicture}
\hspace{5mm}
\begin{tikzpicture}
\startnodes{$\KK$}
\redpath{t1}{r3}
\redpath{t3}{r1}
\redpath{t1}{r2}
\end{tikzpicture}
\caption{All possible states under consideration.}
\label{fig:All_cases}
\end{figure*}
In general, removing the interference might have to be performed over a set of states (JEMS). In addition to this, state splitting is also required. 
The optimality of our proposed scheme is shown by comparing the achievable sum-rate with a genie aided upper bound.
\section{System Model}
Consider a wired network with three transmitters (Tx), which want to communicate with their desired receivers (Rx). Tx\,$i$, $i\in\{1,2,3\}$ wants to send a message $W_i$ to Rx\,$i$ (see Fig.~\ref{fig:System_model}). It encodes this message into a length-$n$ sequence $\X_i= (X_i(1),\ldots,X_i(n))$ \textcolor{C1}{and sends this sequence}. The received symbol at Rx\,${j}$ in \textcolor{c3}{the} $k$th channel use is given by 
\begin{align}
Y_j(k) = \sum_{i=1}^3 h_{ji}(k) X_i(k), \quad \forall j\in\lbrace 1, 2, 3 \rbrace \label{eq:received_symbol},
\end{align}
where $X_i(k)$ and $h_{ji}(k)$ denote the transmitted symbol by Tx\,$i$ and the time-variant channel coefficient corresponding to the link between Tx\,$i$ and Rx\,${j}$, respectively. All symbols are chosen from a Galois Field \textcolor{c3}{($\mathbb{GF}$)}. Moreover, the linear operations are \textcolor{c2}{performed} over this $\mathbb{GF}$. The capacity of each \textcolor{c3}{point to point} channel is $\log|\mathbb{GF}|$, where $|\mathbb{GF}|$ represents the cardinality of $\mathbb{GF}$. Therefore, \textcolor{C1}{only} one symbol can be transmitted over a link \textcolor{C1}{per} channel use. 

In our model, CSIT is restricted only to the topology of the network. Therefore, the only information available \textcolor{C1}{to} the transmitters is about \textcolor{C1}{the} presence \textcolor{C1}{or absence} of links but not about the channel coefficients.\footnote{\textcolor{rev}{In order to make the available CSIT strictly weaker than perfect CSI, we set the cardinality of $\mathbb{GF}$ larger than $2$.}} However, \textcolor{c3}{each Rx knows its} channel coefficients \textcolor{c3}{in addition to} the topology of the network. 

Since the channel coefficients change \textcolor{c3}{over time}, the connectivity of the network varies during the transmission. \textcolor{c3}{Note that each connectivity state occurs with a certain probability.} 



It is worthy to note that the receivers start the decoding after receiving \textcolor{C1}{a complete sequence} $\Y_j$ \textcolor{c3}{as there are no latency constraints}. Therefore, the order of occurrence of the states is not important. \textcolor{c2}{Let $\mathcal{A}$ be a subset of the set of states shown in Fig. \ref{fig:All_cases}} and $\boldsymbol{X}_{i,\mathcal{A}}$ be the sequence of transmitted symbols by Tx\,$i$ in all states in $\mathcal{A}$. Assuming a length-$n$ sequence $\X_i$, \textcolor{c3}{in which $n$ is sufficiently large,} the length of $\boldsymbol{X}_{i,\mathcal{A}}$ is $n \lambda_\mathcal{A}$, where $\lambda_\mathcal{A}$ denotes the sum of the probabilities of the states in $\mathcal{A}$. Note that if the set $\mathcal{A}$ includes one single state, we drop braces for simplicity in writing. For example, $\boldsymbol{X}_{i,A}$ is the sequence of transmitted symbols by Tx\,$i$ in $\mathcal{A} = \{A\}$. \textcolor{rev}{Note that $\lambda_{\{A,B,C,D,E,F,G\}}=1$.}
%

\textcolor{rev}{The probability of error, achievable rates $R_1$, $R_2$, $R_3$, and sum-capacity are defined in the standard Shannon sense \cite{CoverThomas}.}

In the following the joint encoding strategy in~\cite{GherekhlooChaabanSezginIZS14} is briefly introduced in order to motivate that more sophisticated schemes are needed when we are considering the non-equiprobable connectivity states. 

\section{\textcolor{c3}{Joint encoding strategy in TIM}}
To discuss the main idea of joint encoding, consider a network with three Tx and Rx in which the interference links can appear in a cyclic manner shown in Fig.~\ref{fig:cyclic manner}. 
In this example, the network has a total of 8 states. Among all these, there are three states \textcolor{c3}{with a single} interference \textcolor{c3}{link}. We denote these interference states by I-states (see Fig.~\ref{fig:IRstates}). 
\textcolor{c3}{Considering all transmitters to be active, all receivers will get interference in I-states.}
As it is shown in \cite{GherekhlooChaabanSezginIZS14}, we can resolve the interferences at the receivers by using a single state in which all the interference links appear. This resolving state is denoted by R-state in Fig.~\ref{fig:IRstates}. 
Note that no extra interference link (\textcolor{c3}{apart from} the interference links in I-states) appears in the R-state. 
In a network with two users, the R-state is the \textcolor{c3}{state in which both interference links are present (as cross links). By using the cross links in this state the interference caused in Z- and S-channels can be resolved \cite{SunGengJafar}.} 

In general, this joint encoding scheme is not enough to obtain the capacity of three user interference channels. As an example, we consider a subset of all possible connectivity states shown in Fig.~\ref{fig:All_cases}. In what follows, we highlight the key facts which \textcolor{rev}{make the discussed joint encoding scheme non-applicable} to the states in Fig.~\ref{fig:All_cases}. These facts provide valuable insights to study this network.

\newcommand{\startnodescyclic}[7]{
\node at (0.5,0) [above] {#1};
\node (t1) at (0,0) [inner sep=0]{};
\node (t1s) at (0,0) [inner sep=0,left]{#2};
\node (t2) at (0,-1) [inner sep=0]{};
\node (t2s) at (0,-1) [inner sep=0,left]{#3};
\node (t3) at (0,-2) [inner sep=0] {};
\node (t3s) at (0,-2) [inner sep=0,left] {#4};
\node (r1) at (2,0) [inner sep=0] {};
\node (r1s) at (2,0) [inner sep=0,right] {#5};
\node (r2) at (2,-1) [inner sep=0] {};
\node (r2s) at (2,-1) [inner sep=0,right] {#6};
\node (r3) at (2,-2) [inner sep=0] {};
\node (r3s) at (2,-2) [inner sep=0,right] {#7};
\draw[->] (t1) to (r1);
\draw[->] (t2) to (r2);
\draw[->] (t3) to (r3);}
\newcommand{\redpathcyclic}[2]{\draw[dashed,->,red] (#1) to (#2);}
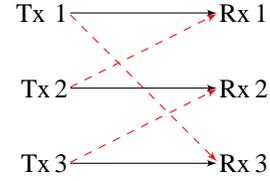
\begin{figure} 
\centering
\begin{tikzpicture}
\startnodescyclic{}{Tx~1}{Tx\,2}{Tx\,3}{Rx\,1}{Rx\,2}{Rx\,3}
\redpathcyclic{t2}{r1}
\redpathcyclic{t3}{r2}
\redpathcyclic{t1}{r3}
\end{tikzpicture}
\caption{The interference links can appear in a cyclic manner.}
\label{fig:cyclic manner}
\end{figure}
\newcommand{\IRstates}[1]{
\node at (0.5,0) [above] {#1};
\node (t1) at (0,0) [inner sep=0] {};
\node (t2) at (0,-0.4) [inner sep=0] {};
\node (t3) at (0,-0.8) [inner sep=0] {};
\node (r1) at (1,0) [inner sep=0] {};
\node (r2) at (1,-0.4) [inner sep=0] {};
\node (r3) at (1,-0.8) [inner sep=0] {};
\draw[->] (t1) to (r1);
\draw[->] (t2) to (r2);
\draw[->] (t3) to (r3);}
\newcommand{\IRredpath}[2]{\draw[->,red] (#1) to (#2);}
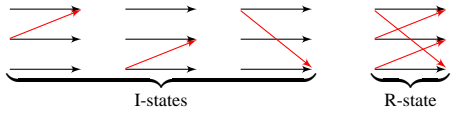
\begin{figure} 
\centering
$\underbrace{\begin{tikzpicture}
\IRstates{}
\IRredpath{t2}{r1}
\end{tikzpicture}
\hspace{5mm}
\begin{tikzpicture}
\IRstates{}
\IRredpath{t3}{r2}
\end{tikzpicture}
\hspace{5mm}
\begin{tikzpicture}
\IRstates{}
\IRredpath{t1}{r3}
\end{tikzpicture}}_{\text{I-states}}$
\hspace{5mm}
$\underbrace{\begin{tikzpicture}
\IRstates{}
\IRredpath{t1}{r3}
\IRredpath{t2}{r1}
\IRredpath{t3}{r2}
\end{tikzpicture}}_{\text{R-state}}$
\caption{The interference links in I-states appear in R-state in a cyclic manner without additional interference link. Therefore, all interferences caused in I-states can be resolved by using a single R-state \cite{GherekhlooChaabanSezginIZS14}.}
\label{fig:IRstates}
\end{figure}
\newcommand{\achievability}[7]{
\node at (0.5,0) [above] {#1};
\node (t1) at (0,0) [inner sep=0]{};
\node (t1s) at (0,0) [inner sep=0,left]{#2};
\node (t2) at (0,-0.4) [inner sep=0]{};
\node (t2s) at (0,-0.4) [inner sep=0,left]{#3};
\node (t3) at (0,-0.8) [inner sep=0] {};
\node (t3s) at (0,-0.8) [inner sep=0,left] {#4};
\node (r1) at (1,0) [inner sep=0] {};
\node (r1s) at (1,0) [inner sep=0,right] {#5};
\node (r2) at (1,-0.4) [inner sep=0] {};
\node (r2s) at (1,-0.4) [inner sep=0,right] {#6};
\node (r3) at (1,-0.8) [inner sep=0] {};
\node (r3s) at (1,-0.8) [inner sep=0,right] {#7};
\draw[->] (t1) to (r1);
\draw[->] (t2) to (r2);
\draw[->] (t3) to (r3);}
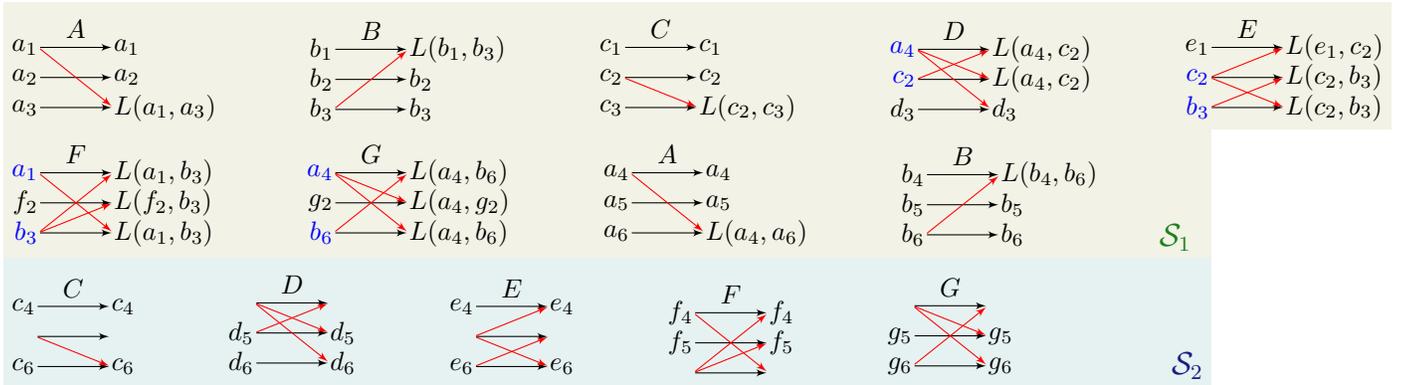
\begin{figure*} 
\hilightone{\begin{tikzpicture}
\achievability{$\D$}{$a_1$}{$a_2$}{$a_3$}{$a_1$}{$a_2$}{$L(a_1,a_3)$}
\redpath{t1}{r3}
\end{tikzpicture}
\hspace{10mm}
\begin{tikzpicture}
\achievability{$\DD$}{$b_1$}{$b_2$}{$b_3$}{$L(b_1,b_3)$}{$b_2$}{$b_3$}
\redpath{t3}{r1}
\end{tikzpicture}
\hspace{10mm}
\begin{tikzpicture}
\achievability{$\CC$}{$c_1$}{$c_2$}{$c_3$}{$c_1$}{$c_2$}{$L(c_2,c_3)$}
\redpath{t2}{r3}
\end{tikzpicture}
\hspace{10mm}
\begin{tikzpicture}
\achievability{$\I$}{\textcolor{res}{$a_4$}}{\textcolor{res}{$c_2$}	}{$d_3$}{$L(a_4,c_2)$}{$L(a_4,c_2)$}{$d_3$}
\redpath{t1}{r3}
\redpath{t2}{r1}
\redpath{t1}{r2}
\end{tikzpicture}
\hspace{10mm}
\begin{tikzpicture}
\achievability{$\J$}{$e_1$}{\textcolor{res}{$c_2$}}{\textcolor{res}{$b_3$}}{$L(e_1,c_2)$}{$L(c_2,b_3)$}{$L(c_2,b_3)$}
\redpath{t2}{r3}
\redpath{t2}{r1}
\redpath{t3}{r2}
\end{tikzpicture}}
\hspace{10mm}
\hilightone{\begin{tikzpicture}
\achievability{$\K$}{\textcolor{res}{$a_1$}}{$f_2$}{\textcolor{res}{$b_3$}}{$L(a_1,b_3)$}{$L(f_2,b_3)$}{$L(a_1,b_3)$}
\redpath{t1}{r3}
\redpath{t3}{r1}
\redpath{t3}{r2}
\end{tikzpicture}
\hspace{10mm}
\begin{tikzpicture}
\achievability{$\KK$}{\textcolor{res}{$a_4$}}{$g_2$}{\textcolor{res}{$b_6$}}{$L(a_4,b_6)$}{$L(a_4,g_2)$}{$L(a_4,b_6)$}
\redpath{t1}{r3}
\redpath{t3}{r1}
\redpath{t1}{r2}
\end{tikzpicture}
\hspace{10mm}
\begin{tikzpicture}
\achievability{$\D$}{$a_4$}{$a_5$}{$a_6$}{$a_4$}{$a_5$}{$L(a_4,a_6)$}
\redpath{t1}{r3}
\end{tikzpicture}
\hspace{10mm}
\begin{tikzpicture}
\achievability{$\DD$}{$b_4$}{$b_5$}{$b_6$}{$L(b_4,b_6)$}{$b_5$}{$b_6$}
\redpath{t3}{r1}
\end{tikzpicture} \quad \quad   \large\textcolor{highlight_S1}{$\mathcal{S}_1$} }
\\
\hilighttwo{\begin{tikzpicture}
\achievability{$\CC$}{$c_4$}{}{$c_6$}{$c_4$}{}{$c_6$}
\redpath{t2}{r3}
\end{tikzpicture}
\hspace{10mm}
\begin{tikzpicture}
\achievability{$\I$}{}{$d_5$}{$d_6$}{}{$d_5$}{$d_6$}
\redpath{t1}{r3}
\redpath{t2}{r1}
\redpath{t1}{r2}
\end{tikzpicture}
\hspace{10mm}
\begin{tikzpicture}
\achievability{$\J$}{$e_4$}{}{$e_6$}{$e_4$}{}{$e_6$}
\redpath{t2}{r3}
\redpath{t2}{r1}
\redpath{t3}{r2}
\end{tikzpicture}
\hspace{10mm}
\begin{tikzpicture}
\achievability{$\K$}{$f_4$}{$f_5$}{}{$f_4$}{$f_5$}{}
\redpath{t1}{r3}
\redpath{t3}{r1}
\redpath{t3}{r2}
\end{tikzpicture}
\hspace{10mm}
\begin{tikzpicture}
\achievability{$\KK$}{}{$g_5$}{$g_6$}{}{$g_5$}{$g_6$}
\redpath{t1}{r3}
\redpath{t3}{r1}
\redpath{t1}{r2}
\end{tikzpicture}\quad \quad  \quad \quad \quad \quad \large\textcolor{highlight_S2}{$\mathcal{S}_2$}}
\caption{$L(\cdot,\cdot)$ operator is a linear operation which depends on the channel coefficients \textcolor{c3}{not known at Tx}. Therefore, the $L$ operator changes over the time. By using joint encoding over the states together with symbol extension, we can transmit 29 symbols over 14 channel uses.}
\label{fig:Achievability2}
\end{figure*}
\begin{itemize}
\item In Fig.~\ref{fig:All_cases}, there is no state in which the cross channel without an extra interference link appears. In other words, there is no single R-state. As an example, the interference links in states $\D$ and $\DD$ appear in $\K$. However, an extra interference link exists in state $\K$, which is the link between Tx\,3 and Rx\,2. This extra interference link in $\K$ together with its complementary interference link (the link between Tx\,2 and Rx\,3) in state $\CC$ appear in state $\J$. Again, in state $\J$, we have an extra interference link between Tx\,2 and Rx\,1 in addition to the cross channel. Due to this entanglement, interference symbols cannot be resolved over a single state (in contrast to the known schemes in \cite{SunGengJafar} and \cite{GherekhlooChaabanSezginIZS14}) but \textit{successively} within a group of states.
\item The interference links in states $\D$ and $\DD$ appear together in states $\K$ and $\KK$. 
\textcolor{c3}{Roughly speaking, both of these states can be used to resolve the interference caused in states $\D$ and $\DD$. In order to fully utilize the benefit associated in resolving over cross channels of both states $\K$ and $\KK$, we need to involve a new pair of $\D$ and $\DD$.} This leads to the idea of joint encoding based on \textit{state splitting}.
\end{itemize}
The exact explanation of the scheme is given in proof of Theorem \ref{Theorem_New_Ach}.
\section{Main Result}
\label{sec:Main_Result}
The following theorem provides the main result of this work.
\begin{theorem}
\textcolor{rev}{The sum-capacity for} the three user interference \textcolor{rev}{wired} channel with alternating connectivity shown in Fig.~\ref{fig:All_cases} is $(2+\lambda)\log|\mathbb{GF}|$, where $$\lambda= \min\left\{\frac{\lambda_\D}{2},\frac{\lambda_\DD}{2},\lambda_\CC,\lambda_\I,\lambda_\J,\lambda_\K,\lambda_\KK\right\}.$$
\label{Theorem_New_Ach}
\end{theorem}	
\begin{proof}
To establish Theorem~\ref{Theorem_New_Ach}, we need to find an optimal achievability scheme. The optimality of the scheme is shown by comparing it with a tight upper bound of the sum-capacity. We start by introducing an achievability scheme leading to a \textcolor{c3}{sum-rate} lower bound denoted \textcolor{c3}{$R_e$}.
\subsection*{Achievability:}	
The achievability scheme is based on JEMS together with state splitting.
\textcolor{c3}{The goal is to resolve the interference signals caused in some states. To do this, as in \cite{SunGengJafar}, we retransmit the interference signals in a state in which the links corresponding to the interference symbols appear.}
In our case, the interference links in states $\D$ and $\DD$ appear together in states $\K$ and $\KK$. The cross channels in states $\K$ and $\KK$ can be used for resolving the interference signal caused in states $\D$ and $\DD$. In order to use the cross channel in both states $\K$ and $\KK$, we need to consider two pairs of $\D$ and $\DD$. \textcolor{c3}{Therefore, we split all states into two parts, and perform joint encoding over 9 states represented by $\mathcal{S}_1$ shown in Fig.~\ref{fig:Achievability2}.} As it is shown in Fig.~\ref{fig:Achievability2}, the interference symbols in states $\D$ and $\DD$ are resolved by using states $\K$ and $\KK$. In both states $\K$ and $\KK$, there is an interference link in addition to the cross channel. 
\textcolor{c3}{Interestingly, these interference links together with their complementary interference links in states $\CC$ and $\J$ appear as cross channels in states $\J$ and $\I$, respectively. Therefore, the interference symbols caused in states $\{\K,\CC\}$ and $\{\KK,\J\}$ are resolved using the cross channels in states $\J$ and $\I$, respectively.} \textcolor{c3}{Now, consider the encoding scheme over the states in $\mathcal{S}_1$, shown in Fig.~\ref{fig:Achievability2}. It is shown that each Rx needs to decode 9 symbols. As an example, Rx1 needs to decode $\{a_1,b_1,b_3,c_1,a_4,c_2,e_1,b_6,b_4\}$.} This is possible due to 9 linear independent equations which are available at each Rx. \textcolor{c3}{Compared to the achievability scheme in \cite{GherekhlooChaabanSezginIZS14}, resolving the interference is not performed over a single state. Therefore, we use JEMS over the states in $\mathcal{S}_1$ to transmit 19 symbols reliably over 9 channel uses.} Note that by considering these states separately, we cannot transmit more than 2 symbols per channel use. 

In order to complete the achievability scheme, we need to consider the states in $\mathcal{S}_2=(\CC, \I, \J, \K, \KK )$ \textcolor{c3}{in addition to $\mathcal{S}_1$}. 
As it is shown in Fig.~\ref{fig:Achievability2}, in all these states, we can transmit 2 symbols reliably per channel use. \textcolor{c3}{Note that the encodings are different for same states (for example $\CC$) in $\mathcal{S}_1$ and $\mathcal{S}_2$.} This fact highlights the notion of state splitting. 
\textcolor{c3}{Since every symbol is chosen from $\mathbb{GF}$ with the entropy $\log|\mathbb{GF}|$, the achievable sum-rates for equiprobable states in $\mathcal{S}_1$ and $\mathcal{S}_2$ are as follows}
\begin{align}
\textcolor{c3}{R_e} =  \begin{cases}
\frac{19}{9}\, \textcolor{c3}{\log|\mathbb{GF}|} \quad &\text{for } \mathcal{S}_1, \\ 
2 \,\textcolor{c3}{\log|\mathbb{GF}|} \quad    &\text{for } \mathcal{S}_2.
\end{cases}\notag 
\end{align}
\textcolor{c3}{Considering $n$ total channel uses, we perform joint encoding across $2n\lambda $ channel uses of states $\{\D,\DD\}$ together with $n\lambda $ channel uses of states $\{\CC,\I,\J,\K,\KK\}$ (see $\mathcal{S}_1$ in Fig.~\ref{fig:Achievability2}). Note that $\lambda$ is defined as $\min\{\frac{\lambda_\D}{2},\frac{\lambda_\DD}{2},\lambda_\CC,\lambda_\I,\lambda_\J,\lambda_\K,\lambda_\KK\}$. Roughly speaking, during the overall transmission, $\mathcal{S}_1$ cannot occur in more channel uses than $n\lambda $. Therefore, we transmit $19n\lambda$ symbols based on the joint encoding over states in $\mathcal{S}_1$ during the overall transmission. In the remaining channel uses, we transmit $2$ symbols. Therefore, the total number of transmitted symbols over $n$ channel uses is given by
\begin{align}
19n\lambda & +  2n  [(\lambda_\D-2\lambda) + (\lambda_\DD-2\lambda)+(\lambda_\CC-\lambda)  \notag \\ 
& + (\lambda_\I-\lambda) +  (\lambda_\J-\lambda)+ (\lambda_\K-\lambda) +  (\lambda_\KK-\lambda)  ] \notag \\
& = n(2+\lambda). \label{eq:Achievability2}
\end{align}
\textcolor{rev}{By dividing \eqref{eq:Achievability2} by the number of channel uses $n$ and multiplying it by $\log|\mathbb{GF}|$ (since the symbol are chosen from $\mathbb{GF}$ with the entropy $\log|\mathbb{GF}|$) we obtain the following achievable sum-rate}
\begin{align}
R_e \leq \left( 2 + \lambda\right)\log|\mathbb{GF}|. \label{eq:LB_new_Ach}
\end{align} }

\subsection*{Upper bound:}
We establish the upper bound \textcolor{rev}{on the sum-capacity} as follows. \textcolor{c3}{Consider the sum-rate}
\begin{align}
n R_\Sigma =&  \sum_{i=1}^{3} H(W_i) \notag \\
=&  \sum_{i=1}^{3} H(W_i) + H(W_i|\Y_i) - H(W_i|\Y_i) \notag \\
\overset{(a)}{\leq} &  \sum_{i=1}^{3} I(W_i;\Y_i) + 3n\epsilon_n, \label{eq:R_sum} 
\end{align}
where ($a$) follows from Fano's inequality \textcolor{c2}{and $\epsilon_n\rightarrow 0$ \textcolor{c3}{for} $n\rightarrow \infty$}. By multiplying the inequality in \eqref{eq:R_sum} by 2, every mutual information appears twice, which corresponds to creating (virtually) three additional receivers. In the next step, we give some side information to the receivers. Therefore, we write
\begin{align}
2n R_\Sigma \leq &   I(W_1;\Y_1,\X_{2,\I},\X_{3,\K},\X_{3,\KK})  \notag \\  
				 & + I(W_2;\Y_2,\X_{3,\J})  + I(W_3;\Y_3)  \notag \\ 
        		  & +  I(W_1;\Y_1,\X_{3,\K},\X_{3,\KK})  \notag \\  
 		  		 & + I(W_2;\Y_2,\X_{1,\I},\X_{3,\J})  \notag \\   
 		  		 & + I(W_3;\Y_3,\X_{1,\K},\X_{1,\KK},\X_{2,\J})  + 6n\epsilon_n.\notag
\end{align}
By using the chain rule and since the messages of three transmitters are independent from each other, we write
\begin{align}
2n R_\Sigma \leq &   I(W_1;\Y_1|\X_{2,\I},\X_{3,\K},\X_{3,\KK})  \notag \\  
				 & + I(W_2;\Y_2|\X_{3,\J})  + I(W_3;\Y_3)  \notag \\ 
        		  & +  I(W_1;\Y_1|\X_{3,\K},\X_{3,\KK})  \notag \\  
 		  		 & + I(W_2;\Y_2|\X_{1,\I},\X_{3,\J})  \notag \\   
 		  		 & + I(W_3;\Y_3|\X_{1,\K},\X_{1,\KK},\X_{2,\J})  + 6n\epsilon_n.\label{eq:NA_4_1_1}
\end{align}
By expressing the mutual information as entropy terms, \eqref{eq:NA_4_1_1} is restated as
\begin{align}
2n R_\Sigma \leq &    H(\Y_1|\X_{2,\I},\X_{3,\K},\X_{3,\KK})\notag \\ 
				 &   -H(\Y_1|\X_{2,\I},\X_{3,\K},\X_{3,\KK},W_1)  \notag \\  
				 &   +H(\Y_2|\X_{3,\J})  -  H(\Y_2|\X_{3,\J},W_2) + H(\Y_3) \notag \\ 
				 &  - H(\Y_3|W_3) + H(\Y_1|\X_{3,\K},\X_{3,\KK}) \notag \\  
 		  		 &  - H(\Y_1|\X_{3,\K},\X_{3,\KK},W_1) + H(\Y_2|\X_{1,\I},\X_{3,\J}) \notag \\   
 		  		 &  - H(\Y_2|\X_{1,\I},\X_{3,\J},W_2) \notag \\    
 		  		 &  + H(\Y_3|\X_{1,\K},\X_{1,\KK},\X_{2,\J}) \notag \\
 		  		 &  - H(\Y_3|\X_{1,\K},\X_{1,\KK},\X_{2,\J},W_3)  + 6n\epsilon_n.\label{eq:NA_4_1}
\end{align}
By considering \eqref{eq:received_symbol} together with the fact that channel coefficients are completely known at the destination, the expression in \eqref{eq:NA_4_1} can be rewritten as in \eqref{eq:NA_4_2} on the top of the next page. By using the chain rule, together with the facts that conditioning does not increase entropy, and that the messages of the users are independent of each other, the individual terms in \eqref{eq:NA_4_2} can be
rewritten as in \eqref{eq:NA_4_7}-\eqref{eq:NA_4_8}. 
\begin{figure*}[t]
\begin{align}
2n R_\Sigma \leq &    H(\X_{1,\{\D,\CC,\I,\K,\KK\}},\Y_{1,\{\DD,\J\}}|\X_{2,\I},\X_{3,\{\K,\KK\}}) 				        				-H(\X_{3,\DD},\X_{2,\J}|\X_{2,\I},\X_{3,\{\K,\KK\}},W_1)  \notag \\  
				 &   +H(\X_{2,\{\D,\DD,\CC,\J\}},\Y_{2,\{\I,\K,\KK\}}|\X_{3,\J})  -  H(\X_{1,\{\I,\KK\}},\X_{3,\K}|\X_{3,\J},W_2) 		 				 \notag \\ &
				  + H(\Y_3)   - H(\X_{1,\{\D,\I,\K,\KK\}},\X_{2,\{\CC,\J\}}|W_3) \notag \\
				 & + H(\X_{1,\{\D,\CC,\K,\KK\}},\Y_{1,\{\DD,\I,\J\}}|\X_{3,\{\K,\KK\}}) - H(\X_{2,\{\I,\J\}},\X_{3,\DD}|\X_{3,\{\K,\KK\}},W_1)  \notag \\ 
				 &+ H(\X_{2,\{\D,\DD,\CC,\I,\J\}},\Y_{2,\{\K,\KK\}}|\X_{1,\I},\X_{3,\J}) - H(\X_{3,\K},\X_{1,\KK}|\X_{1,\I},\X_{3,\J},W_2) \notag \\  
 		  		 &  + H(\X_{3,\{\DD,\J,\K,\KK\}},\Y_{3,\{\D,\CC,\I\}}|\X_{1,\{\K,\KK\}},\X_{2,\J})  - H(\X_{1,\{\D,\I\}}, \X_{2,\CC}|\X_{1,\{\K,\KK\}},\X_{2,\J},W_3) + 6n\epsilon_n \label{eq:NA_4_2} 
\end{align}
\hrule
\begin{align}
  H(\X_{1,\{\D,\CC,\I,\K,\KK\}},\Y_{1,\{\DD,\J\}}|\X_{2,\I},\X_{3,\K},\X_{3,\KK}) \leq & H(\X_{1,\{\D,\I,\K,\KK\}}) + H(\X_{1,\CC},\Y_{1,\{\DD,\J\}}) \label{eq:NA_4_7} \\
H(\X_{3,\DD},\X_{2,\J}|\X_{2,\I},\X_{3,\{\K,\KK\}},W_1)   = & H(\X_{2,\J}|\X_{2,\I})+H(\X_{3,\DD}|\X_{3,\{\K,\KK\}})  \\
H(\X_{2,\{\D,\DD,\CC,\J\}},\Y_{2,\{\I,\K,\KK\}}|\X_{3,\J}) \leq & H(\X_{2,\{\CC,\J\}}) + H(\X_{2,\{\D,\DD\}},\Y_{2,\{\I,\K,\KK\}}) \\
H(\X_{1,\{\I,\KK\}},\X_{3,\K}|\X_{3,\J},W_2)  = &  H(\X_{3,\K}|\X_{3,\J}) + H(\X_{1,\KK}) + H(\X_{1,\I}|\X_{1,\KK}) \\ 
 H(\X_{1,\{\D,\I,\K,\KK\}},\X_{2,\{\CC,\J\}}|W_3)   = &  H(\X_{2,\{\J,\CC\}}) + H(\X_{1,\{\D,\I,\K,\KK\}}) \\
H(\X_{1,\{\D,\CC,\K,\KK\}},\Y_{1,\{\DD,\I,\J\}}|\X_{3,\{\K,\KK\}}) \leq & H(\X_{1,\KK}) + H(\X_{1,\K}|\X_{1,\KK}) + H(\X_{1,\D}|\X_{1,\{\KK,\K\}})  \notag \\ & + H(\Y_{1,\{\I,\J,\DD\}},\X_{1,\CC}) \\
 H(\X_{2,\{\I,\J\}},\X_{3,\DD}|\X_{3,\{\K,\KK\}},W_1)  = & H(\X_{3,\DD}|\X_{3,\{\K,\KK\}}) + H(\X_{2,\{\I,\J\}})  \\
 H(\X_{2,\{\D,\DD,\CC,\I,\J\}},\Y_{2,\{\K,\KK\}}|\X_{1,\I},\X_{3,\J}) \leq & H(\X_{2,\{\I,\J\}}) + H(\X_{2,\CC}|\X_{2,\J}) + H(\X_{2,\{\D,\DD\}},\Y_{2,\{\K,\KK\}}) \\ 
 H(\X_{3,\K},\X_{1,\KK}|\X_{1,\I},\X_{3,\J},W_2)  = &  H(\X_{3,\K}|\X_{3,\J})+ H(\X_{1,\KK}|\X_{1,\I}) \\
 H(\X_{3,\{\DD,\J,\K,\KK\}},\Y_{3,\{\D,\CC,\I\}}|\X_{1,\{\K,\KK\}},\X_{2,\J})  \leq & H(\X_{3,\J}) + H(\X_{3,\K}|\X_{3,\J}) + H(\X_{3,\KK}|\X_{3,\{\J,\K\}}) \notag \\ &+ H(\X_{3,\DD}|\X_{3,\{\K,\KK\}}) + H(\Y_{3,\{\D,\I,\CC\}}) \\
  H(\X_{1,\{\D,\I\}},\X_{2,\CC}|\X_{1,\{\K,\KK\}},\X_{2,\J},W_3) = &  H(\X_{1,\D}|\X_{1,\{\K,\KK\}}) + H(\X_{1,\I}|\X_{1,\{\K,\KK,\D\}})+ H(\X_{2,\CC}|\X_{2,\J}) \label{eq:NA_4_8}
\end{align}
\hrule

\end{figure*}
We can see that by substituting \textcolor{c3}{the right hand side of} \eqref{eq:NA_4_7}-\eqref{eq:NA_4_8} into \eqref{eq:NA_4_2} many terms will cancel out \textcolor{rev}{(it follows directly after summing up \eqref{eq:NA_4_7}-\eqref{eq:NA_4_8})} and we can rewrite \eqref{eq:NA_4_2} as
\begin{align}
2nR_\Sigma  \leq &  H(\X_{1,\CC},\Y_{1,\{\DD,\J\}}) + H(\X_{2,\{\D,\DD\}},\Y_{2,\{\I,\K,\KK\}}) \notag \\ 
				 & + H(\Y_3) + H(\X_{1,\K}|\X_{1,\KK})  + H(\Y_{1,\{\I,\J,\DD\}},\X_{1,\CC}) \notag \\
				 & + H(\X_{2,\{\D,\DD\}},\Y_{2,\{\K,\KK\}})  + H(\X_{3,\J})  \notag \\ 
				 & + H(\X_{3,\KK}|\X_{3,\{\J,\K\}}) + H(\Y_{3,\{\D,\I,\CC\}}) + 6n\epsilon_n. \label{eq:NA_4_9}
\end{align}
After dropping \textcolor{c3}{the conditioning in all terms}, the inequality \eqref{eq:NA_4_9} can be upper bounded by the following expression
\begin{align}
2nR_\Sigma  \leq & \textcolor{rev}{n}\log|\mathbb{GF}| [\lambda_{\{\CC,\DD,\J\}} + \lambda_{\{\D,\DD,\I,\K,\KK\}} + 1 + \lambda_{\K} \notag \\ 
				 & +\lambda_{\{\I,\J,\DD,\CC\}} + \lambda_{\{\D,\DD,\K,\KK\}} + \lambda_{\J}  \notag \\ 
				 & + \lambda_{\KK} + \lambda_{\{\D,\I,\CC\}}] + 6n\epsilon_n, \label{eq:NA_4_10}
\end{align}
where we used the chain rule, the facts that conditioning does not increase the entropy, and that the entropy of discrete random variable in $\mathbb{GF}$ is upper bounded by $\log|\mathbb{GF}|$ \cite{CoverThomas}. Next, we divide the inequality in \eqref{eq:NA_4_10} by $2n$, and let $n\rightarrow \infty$. Then, we obtain
\begin{align}
R_\Sigma  \leq & \log|\mathbb{GF}|\left(2 + \frac{\lambda_\DD}{2}\right). \label{eq:NA_4_11}
\end{align}
\textcolor{c3}{Using the same technique with appropriate genie information and smart utilizations of chain rule, we can also establish the following upper bound
\begin{align}
R_\Sigma  \leq & \log|\mathbb{GF}|\left(2 + \min\left\lbrace\frac{\lambda_\D}{2}, \lambda_\CC, \lambda_\I, \lambda_\J ,\lambda_\K , \lambda_\KK\right\rbrace\right). \label{eq:NA_4_12}
\end{align}
Due to the page limitation, we do not present the proofs. Considering all bounds in \eqref{eq:NA_4_11}, \eqref{eq:NA_4_12}, we get 
\begin{align}
R_\Sigma  \leq & \log|\mathbb{GF}|\left(2 + \lambda\right), \label{eq:NA_4_13}
\end{align}}
which agrees with the lower bound in \eqref{eq:LB_new_Ach} and completes the proof of Theorem~\ref{Theorem_New_Ach}. 
\end{proof}
\section{Conclusion}
We studied the sum-capacity of the three users interference wired channel with an alternating connectivity with only topological knowledge at the transmitters. It is assumed that the connectivity states are non-equiprobable. \textcolor{rev}{As shown in \cite{Jafar}, this result translates to a DoF result for the corresponding wireless network.} We proposed a new achievability scheme which is based on joint encoding across connectivity states and decoding using a multiple interference resolving states (JEMS) combined with splitting the connectivity states. By establishing a genie aided upper bound, the optimality of transmission scheme is shown. 

\balance
\bibliography{myBib}
\end{document}